\documentclass{article}
\usepackage{amssymb,amsmath,amsthm,amsfonts}
\usepackage[ruled,vlined]{algorithm2e}
\usepackage{framed,fullpage,nicefrac,graphicx,comment}
\usepackage{tikz}
\usetikzlibrary{shapes}
\usetikzlibrary{shapes.callouts}
\usetikzlibrary{decorations}
\usetikzlibrary{arrows,decorations.pathmorphing,decorations.shapes}

\newcommand{\R}{{\mathbb R}}

\newcommand{\F}{{\mathbb F}}

\newcommand{\C}{{\mathbb C}}

\newcommand{\PR}[1]{{\mathbb{P}}\left\{ #1\right\}}

\newcommand{\EE}{\mathbb{E}}

\newcommand{\inabs}[1]{\left|#1\right|}

\newcommand{\ind}[1]{\ensuremath{\mathbf{1}_{#1}}}

\newcommand{\conv}{\operatorname{conv}}

\newcommand{\poly}{\mathrm{poly}}

\newcommand{\rank}{\ensuremath{\operatorname{rank}}}

\newcommand{\ip}[2]{\ensuremath{\left\langle #1,#2\right\rangle}}
\newcommand{\inset}[1]{\left\{#1\right\}}
\newcommand{\inparen}[1]{\left(#1\right)}

\newcommand{\suchthat}{\,:\,}

\newcommand{\eps}{\varepsilon}
\newcommand{\vphi}{\varphi}

\newtheorem{theorem}{Theorem} 
\newtheorem{lemma}[theorem]{Lemma} 
\newtheorem{definition}[theorem]{Definition} 
 
\newtheorem{obs}[theorem]{Observation} 
 
\newtheorem{cor}[theorem]{Corollary} 

\newtheorem{remark}{Remark}

\newtheorem{proposition}[theorem]{Proposition}

\newcommand{\Arow}{f}
\newcommand{\Cproblemma}{C_0}
\newcommand{\Cvar}{C_1}
\newcommand{\Clt}{C_2}
\newcommand{\Ctmp}{C_3}
\newcommand{\bigrho}{\ensuremath{\inparen{ 1 - \nicefrac{1}{q} }\inparen{ 1- \eps } } }
\newcommand{\RM}{\mathrm{RM}}
\title{On the list decodability of random linear codes with large error rates}
\author{Mary Wootters
\thanks{University of Michigan, Ann Arbor. \texttt{wootters@umich.edu}. This work was supported by NSF CCF-1161233.}}
\begin{document}
\maketitle

\begin{abstract}
It is well known that a random $q$-ary code of rate $\Omega(\eps^2)$ is list decodable up to radius $(1 - 1/q - \eps)$ with list sizes on the order of $1/\eps^2$, with probability $1 - o(1)$.  However, until recently, a similar statement about random \em linear \em codes has until remained elusive.  In a recent paper, Cheraghchi, Guruswami, and Velingker show a connection between list decodability of random linear codes and the Restricted Isometry Property from compressed sensing, and use this connection to prove that a random linear code of rate $\Omega( \eps^2 /\log^3(1/\eps))$ achieves the list decoding properties above, with constant probability.  We improve on their result to show that in fact we may take the rate to be $\Omega(\eps^2)$, which is optimal, and further that the success probability is $1 - o(1)$, rather than constant.  As an added benefit, our proof is relatively simple.  
Finally, we extend our methods to more general ensembles of linear codes.  As an example, we show that randomly punctured Reed-Muller codes have the same list decoding properties as the original codes, even when the rate is improved to a constant.
\end{abstract}

\section{Introduction}

In the theory of error correcting codes, one attempts to obtain subsets (codes) $\mathcal{C} \subset [q]^n$ which are simultaneously large and ``spread out."  If the rate of the code $R = \log_q|\mathcal{C}|/n$ is large, then each codeword $c \in \mathcal{C}$ contains a large amount of information.  On the other hand, if the distance between any two codewords is large, then even if a codeword becomes corrupted, say, a fraction $\rho$ of its entries are changed, the original codeword may be uniquely recovered.
There is a trade-off between the rate and distance, and sometimes this trade-off can be too harsh: it is not always necessary to recover exactly the intended codeword $c$, and sometimes suffices to recover a short list of $L$ codewords.  This relaxed notion, called \em list decoding, \em  was introduced in the 1950's by Elias~\cite{elias} and Wozencraft~\cite{wozen}.  More formally, a code $\mathcal{C}$ is \em $(\rho,L)$-list decodable \em if, for any received word $w$, there are at most $L$ other codewords within relative distance $\rho$ of $w$.

We will be interested in the list decodability of random codes, and in particular random linear codes.  A \em linear code \em of rate $R$ in $\F^n_q$ is a code which forms a linear subspace of $\F^n_q$ of dimension $k = Rn$.  Unless otherwise noted, a \em random \em linear code of rate $R$ will be a uniformly random linear code, where $\mathcal{C}$ is a uniformly random $k$-dimensional linear subspace of $\F_q^n$.

Understanding the trade-offs in list decoding is interesting not just for communication, but also for a wide array of applications in complexity theory.  
List decodable codes can be used for hardness amplification of boolean functions and for constructing hardcore predicates from one-way functions,
and they can be used to construct randomness extractors, expanders, and pseudorandom generators.  (See the surveys~\cite{sudan,vadhan} for these and many more applications).  
Understanding the behavior of linear codes, and in particular random linear codes, is also of interest: decoding a random linear code is related to they problem of learning with errors, a fundamental problem in both learning theory~\cite{bkw,fgkp} and cryptography~\cite{regev:05}.

In this work, we show that for large error rates $\rho$, a random linear code has the optimal list decoding parameters, improving upon the recent result of Cheraghchi, Guruswami, and Velingker~\cite{che:2012}.  Our result establishes the existence of such codes, previously unknown for $q > 2$.  We extend our results to other (not necessarily uniform) ensembles of linear codes, including random families obtained from puncturing Reed-Muller codes.

\subsection{Related Work}
In this paper, we will be interested in large error rates $\rho = \bigrho$, for small $\eps$.
Since a random word $r \in \F_q^n$ will disagree with any fixed codeword on a $1 - 1/q$ fraction of symbols in expectation, this is the largest error rate we can hope for.
This large-$\rho$ regime is especially of interest for applications in complexity theory, so we seek to understand the trade-offs between the achievable rates and list sizes, in terms of $\eps$.

When $\rho$ is \em constant\em, Guruswami, H{\aa}stad, and Kopparty \cite{ghk11} show that a random linear code of
rate $1 - H_q(\rho) - C_{\rho, q}/L$ is $(\rho, L)$-list decodable, where
$H_q(x) = x\log_q(q-1) - x\log_q(x) - (1 - x)\log_q(1-x)$ is the $q$-ary
entropy.  This matches lower bounds of Rudra and Guruswami-Narayanan \cite{rud11, GN12}.  
However, for $\rho = \bigrho$, the constant $C_{\rho, q}$ depends exponentially on $\eps$, and this result quickly degrades. 

When $\rho = \bigrho$, it follows from a straightforward computation that a random (not necessarily linear) code of rate $\Omega(\eps^2)$ is $\inparen{\bigrho, O(1/\eps^2)}$-list decodable.
However, until recently, the best upper bounds known for random linear codes with rate $\Omega(\eps^2)$ had list sizes exponential in $1/\eps$~\cite{zp82}; closing this exponential gap between random linear codes and general random codes was posed by~\cite{elias91}.
The existence of a binary linear code with rate $\Omega(\eps^2)$ and list size $O(1/\eps^2)$ was shown in~\cite{ghsz:2002}.  However, this result only holds for binary codes, and further the proof does not show that most linear codes have this property.
Cheraghchi, Guruswami, and Velingker (henceforth CGV) recently made substantial progress on closing the gap between random linear codes and general random codes.  Using a connection between list decodability of random linear codes and the Restricted Isometry Property (RIP) from compressed sensing, they proved the following theorem. 
\begin{theorem}\label{thm:theirthm}[Theorem 12 in \cite{che:2012}]
Let $q$ be a prime power, and let $\eps, \gamma > 0$ be constant parameters. Then for
all large enough integers $n$, a random linear code $\mathcal{C} \subseteq \F_q^n$ of rate $R$, for some
\[ R \geq C \frac{ \eps^2 }{ \log(1/\gamma) \log^3(q/\eps) \log(q) } \]
is $(\bigrho, O(1/\eps^2))$-list decodable with probability at least $1 - \gamma$.
\end{theorem}
It is known that the
rate cannot exceed $O(\eps^2)$ (this follows from the list decoding capacity theorem).  Further, the recent lower bounds of Guruswami and Vadhan~\cite{gv10}
and Blinovsky~\cite{bli05,bli08} show that the list size $L$ must be at least
$\Omega_q(1/\eps^2)$.
Thus, Theorem \ref{thm:theirthm} has nearly optimal dependence on $\eps$, leaving a polylogarithmic gap.

\subsection{Our contributions}
The extra logarithmic factors in the result of CGV stem from the difficulty in proving that the RIP is likely to hold for randomly subsampled Fourier matrices.  Removing these logarithmic factors is considered to be a difficult problem.  In this work, we show that while the RIP is a sufficient condition for list decoding, it may not be necessary.  We formulate a different sufficient condition for list decodability: while the RIP is about controlling the $\ell_2$ norm of $\Phi x$, for a matrix $\Phi$ and a sparse vector $x$ with $\|x\|_2 = 1$, our sufficient condition amounts to controlling the $\ell_1$ norm of $\Phi x$ with the same conditions on $x$. 
Next, we show, using techniques from high dimensional probability, that this condition does hold with overwhelming probability for random linear codes, with no extra logarithmic dependence on $\eps$.
The punchline, and our main result, is the following theorem.
\begin{theorem}\label{thm:mainthm}
Let $q$ be a prime power, and fix $\eps > 0$. Then for
all large enough integers $n$, a random linear code $\mathcal{C} \subseteq \F_q^n$ of rate $R$, for 
\[ R \geq C \frac{ \eps^2 }{ \log(q) } \]
is $( \bigrho, O(1/\eps^2))$-list decodable with probability at least $1 - o(1)$.
Above, $C$ is an absolute constant.
\end{theorem}
There are three differences between Theorem \ref{thm:theirthm} and Theorem \ref{thm:mainthm}.  First, the dependence on $\eps$ in Theorem \ref{thm:mainthm} is optimal.  Second, the dependence on $q$ is also improved by several log factors.  Finally, the success probability in Theorem \ref{thm:mainthm} is $1 - o(1)$, compared to a constant success probability in Theorem \ref{thm:theirthm}.
As an additional benefit, the proof on Theorem \ref{thm:mainthm} is relatively short, while the proof of the RIP result in~\cite{che:2012} is quite difficult. 


To demonstrate the applicability of our techniques, we extend our approach to
apply to not necessarily uniform ensembles of linear codes.  We formulate a
more general version of Theorem \ref{thm:mainthm}, and give examples of codes
to which it applies.  Our main example is linear codes $\mathcal{E}$ of
rate $\Omega(\eps^2)$ whose generator matrix is chosen by randomly sampling the columns of 
a generator matrix of a linear code $\mathcal{C}$ of nonconstant rate.
Ignoring details about
repeating columns, $\mathcal{E}$ can be viewed as randomly punctured version
of $\mathcal{C}$.  
Random linear codes fit into this framework when $\mathcal{C}$ is taken to be $\RM_q(1,k)$,
the $q$-ary Reed-Muller code of degree one and dimension $k$.
We extend this in a natural way by taking $\mathcal{C} =
\RM(r,m)$ to be any (binary) Reed-Muller code.  It has recently been
shown~\cite{gkz,klp} that $\RM(r,m)$ is list-decodable up to $1/2 - \eps$, with
exponential but nontrivial list sizes.  However, $\RM(r,m)$ is not a ``good"
code, in the sense that it does not have constant rate.  In the same spirit as our
main result, we show that when $\RM(r,m)$ is punctured down to rate $O(\eps^2)$,
with high probability the resulting code is list decodable up to 
radius $1/2 - \eps$ with asymptotically no loss in list size.

\subsection{Our approach}
The CGV proof of Theorem \ref{thm:theirthm} proceeds in three steps.  
The first step is to prove an average case Johnson bound---that is, a sufficient condition for list decoding that depends on the \em average \em pairwise distances between codewords, rather than the worst-case differences.  
The second step is a translation of the coding theory setting to a setting suitable for the RIP: a code $\mathcal{C}$ is encoded as a matrix $\Phi$ whose columns correspond to codewords of $\mathcal{C}$.  This encoding has the property that if $\Phi$ had the RIP with good parameters, then $\mathcal{C}$ is list decodable with similarly good parameters.  Finally, the last and most technical step is proving that the matrix $\Phi$ does indeed have the Restricted Isometry Property with the desired parameters.

In this work, we use the second step from the CGV analysis (the encoding from codes to matrices), but we bypass the other steps.  While both the average case Johnson bound and the improved RIP analysis for Fourier matrices are clearly of independent interest, our analysis will be much simpler, and obtains the correct dependence on $\eps$. 

\subsection{Organization}
In Section \ref{sec:setup}, we fix notation and definitions, and also introduce the simplex encoding map from the second step of the CGV analysis.  In Section \ref{sec:sufficient}, we state our sufficient condition and show that it implies list decoding, which is straightforward.  We take a detour in Section \ref{ssec:RIP} to note that the sufficiency of our condition in fact implies the sufficiency of the Restricted Isometry Property directly, providing an alternative proof of Theorem 11 in~\cite{che:2012}.  In Section \ref{sec:tech} we prove that our sufficient condition holds, and conclude Theorem \ref{thm:mainthm}.  Finally, in Section \ref{sec:gen}, we discuss the generality of our result, and show that it applies to other ensembles of linear codes. 

\section{Definitions and Preliminaries}
\label{sec:setup}
Throughout, we will be interested in linear, $q$-ary, codes $\mathcal{C}$ with length $n$ and size $|\mathcal{C}| = N$. 
We use the notation $[q] = \{0,\ldots, q-1\}$, and for a prime power $q$, $\F_q$ denotes the finite field with $q$ elements.
Nonlinear codes use the alphabet $[q]$, and linear codes use the alphabet $\F_q$.  When notationally convenient, we identify $[q]$ with $\F_q$; for our purposes, this identification may be arbitrary.
We let $\omega = e^{2\pi \mathbf{i} / q}$  denote the primitive $q^{th}$ root of unity, and we use $\Sigma_L \subset \{0,1\}^N$ to denote the space of $L$-sparse binary vectors.
For two vectors $x,y \in [q]^n$, the \em relative Hamming distance \em between them is
\[ d(x,y) = \frac{1}{n} \left| \inset{ i \suchthat x_i \neq y_i } \right|.\]
Throughout, $C_i$ denotes numerical constants.  For clarity, we have made no attempt to
optimize the values of the constants.

A code is list decodable if any received word $w$ does not have too many codewords close to it:
\begin{definition} \label{def:ld}
A code $\mathcal{C} \subseteq [q]^n$ is $(\rho, L)$-list decodable if for all $w \in [q]^n$, 
\[ \left| \inset{ c \in \mathcal{C} \suchthat d(c,w) \leq \rho } \right| \leq L.\]
\end{definition}
A code is \em linear \em if the set
 $\mathcal{C}$ of codewords is of the form $\mathcal{C} = \{ xG \mid x \in \F_q^k \}$,
for a $k \times n$ generator matrix $G$.  
We say that $\mathcal{C}$ is a \em random linear code of rate $R$ \em
if the image of the generator matrix $G$ is a random subspace of dimension $k = Rn$.

Below, it will be convenient to work with generator matrices $G$ chosen uniformly at random from $\F_q^{k \times n}$, rather than with random linear subspaces of dimension $k$.  These are not the same, as there is a small but positive probability that $G$ chosen this way will not have full rank.  However, we observe that
\begin{equation}\label{eq:rank}
 \PR{ \rank(G) < k } = \prod_{r=0}^{k-1}\inparen{1 - q^{r-n}} = 1 - o(1).
\end{equation}
Now suppose that $\mathcal{C}$ is a random linear code of rate $R = k/n$, and $\mathcal{C'}$ is a code with a random $k \times n$ generator matrix $G$.  Let $\mathcal{E}$ be the event that $\mathcal{C}$ is $(\rho,L)$-list decodable for some $\rho$ and $L$, and let $\mathcal{E}'$ be the corresponding event for $\mathcal{C}'$.  By symmetry, we have
\begin{align*}
\PR{ \mathcal{E}} &= \PR{ \mathcal{E}' \mid \rank(G) = k }\\
&\geq \PR{ \mathcal{E}' \wedge \rank(G) = k }\\
&\geq 1 - \PR{ \overline{\mathcal{E}'}} - \PR{ \rank(G) < k }\\
&= \PR{\mathcal{E}'} - o(1),
\end{align*}
where we have used \eqref{eq:rank} in the final line.
Thus, to prove Theorem \ref{thm:mainthm}, it suffices to show that $\mathcal{C}'$ is list decodable, and so 
going forward we will consider a code $\mathcal{C}$ with a random $k \times n$ generator matrix.  For notational convenience, we will also treat $\mathcal{C} = \inset{ xG \mid x \in \F_q^k }$ as a multi-set, so that in particular we always have $N = |\mathcal{C}| = q^k$.  Because by the above analysis the parameter of interest is now $k$, not $|\mathcal{C}|$, this will be innocuous.

We make use the simplex encoding used in the CGV analysis, which maps the code $\mathcal{C}$ to a complex matrix $\Phi$.
\begin{definition}[Simplex encoding from~\cite{che:2012}] 
\ \newline 
Define a map $\vphi:[q] \to \C^{q-1}$ by $\vphi(x)(\alpha) = \omega^{x\alpha}$ for $\alpha \in \{1,\ldots,q-1\}$.  We extend this map to 
a map $\vphi:[q]^n \to \C^{n(q-1)}$ in the natural way by concatenation.
Further, we extend $\vphi$ to act on sets $\mathcal{C} \subset [q]^n$: $\vphi(\mathcal{C})$ is the $n(q-1) \times N$ matrix whose columns are $\vphi(c)$ for $c \in \mathcal{C}$.
\end{definition}

Suppose that $\mathcal{C}$ is a $q$-ary linear code with random generator matrix $G \in \F_q^{k \times n}$, as above.
Consider the $n \times N$ matrix $M$ which has the codewords as columns.  The rows of this matrix are independent---each row corresponds to a column $t$ of the random generator matrix $G$.  To sample a row $r$, we choose $t \in \F_q^k$ uniformly at random (with replacement), and let $r = (\ip{t}{x})_{x \in \F_q^k}$. 
Let $T$ denote the random multiset with elements in $\F_q^k$ consisting of the draws $t$.
To obtain $\Phi = \vphi(\mathcal{C})$, we replace each symbol $\beta$ of $M$ with its simplex encoding $\vphi(\beta)$, regarded as a column vector.  
Thus, each row of $\Phi$ corresponds to a vector $t \in T$ (a row of the original matrix $M$, or a column of the generator matrix $G$), and an index $\alpha \in \{1,\ldots,q-1\}$ (a coordinate of the simplex encoding).  We denote this row by $\Arow_{t,\alpha}$.  

We use the following facts about the simplex encoding, also from~\cite{che:2012}:
\begin{enumerate}
	\item For $x,y \in [q]^n$, 
		\begin{equation}\label{eq:ip}
		 \ip{\vphi(x)}{\vphi(y)} = (q-1)n - q d(x,y) n.
		\end{equation}
	\item If $\mathcal{C}$ is a linear code with a uniformly random generator matrix, the columns of $\Phi$ are orthogonal in expectation.  That is, for $x,y \in \F_q^n$, indexed by $i,j \in \F_q^k$ respectively,
		we have 
		\begin{align*}
		\EE d(x,y) &= \frac{1}{n} \EE \sum_{t \in T} \ind{\ip{t}{i} \neq \ip{t}{j}} \\
		&= \PR{ \ip{t}{i} \neq \ip{t}{j}} \\
		&= \begin{cases} 1 - \frac{1}{q} & i\neq j \\ 0 & i = j
		\end{cases}
		\end{align*}
		Combined with \eqref{eq:ip}, we have
		\begin{align*}
			\EE \ip{\vphi(x)}{\vphi(y)} &= (q-1)n - qn \,\EE d(x,y) \\
			&=  \begin{cases} (q-1)n & x=y \\0 &x \neq y \end{cases}
		\end{align*}
		This implies that
		\begin{equation}\label{eq:expect}
			\EE \| \Phi x \|_2^2 = \sum_{i,j \in [N] } x_ix_j \EE \ip{\vphi(c_i)}{\vphi(c_j)} = (q-1)n\|x\|^2.
		\end{equation}
\end{enumerate}

\section{Sufficient conditions for list decodability}
\label{sec:sufficient}
Suppose that $\mathcal{C}$ is a linear code as above, and let $\Phi = \vphi(\mathcal{C}) \in \C^{n(q-1) \times N}$ be the complex matrix associated with $\mathcal{C}$ by the simplex encoding.
We first translate Definition \ref{def:ld} into a linear algebraic statement about $\Phi$.  The identity \eqref{eq:ip} implies that $\mathcal{C}$ is $(\rho, L-1)$ list decodable if and only if for all $w \in \F_q^n$, for all sets $\Lambda \subset \mathcal{C}$ with $|\Lambda| = L$, there is at least one codeword $c \in \Lambda$ so that $d(w,c) > \rho$, that is, so that
\[ \ip{ \vphi(c) }{ \vphi(w) } < (q-1)n - q \rho n.\]
Translating the quantifiers into appropriate max's and min's, we observe
\begin{obs}\label{obs:char}
A code $\mathcal{C} \in [q]^n$ is $(\rho, L-1)$-list decodable if and only if
\[ \max_{w \in [q]^n} \max_{\Lambda \subset \mathcal{C}, |\Lambda| = L } \min_{c \in \Lambda} \ip{ \vphi(w)}{ \vphi(c) } < (q-1)n - q \rho n. \]
When $\rho = \bigrho$, $\mathcal{C}$ is $(\rho, L-1)$-list decodable if and only if
\begin{equation}\label{eq:star}
 \max_{w \in [q]^n} \max_{\Lambda \subset \mathcal{C}, |\Lambda| = L } \min_{c \in \Lambda} \ip{ \vphi(w)}{ \vphi(c) } < (q-1)n\eps.
\end{equation}
\end{obs}

We seek sufficient conditions for \eqref{eq:star}.  Below is the one we will find useful:
\begin{lemma}\label{lem:sufficient}
Let $\mathcal{C} \in \F_q^n$ be a $q$-ary linear code, and let $\Phi = \vphi(\mathcal{C})$ as above. Suppose that 
\begin{equation}\label{eq:doublestar}
	\frac{1}{L} \max_{x \in \Sigma_{L}} \|\Phi x\|_1 < (q-1) n\eps.
\end{equation}
Then \eqref{eq:star} holds, and hence $\mathcal{C}$ is $(\bigrho, L-1)$-list decodable.
\end{lemma}
\begin{proof}
We always have 
\[\min_{c \in \Lambda} \ip{\vphi (w)}{\vphi(c)} \leq \frac{1}{L} \sum_{c \in \Lambda} \ip{\vphi(w)}{\vphi(c)},\] 
so
\begin{align*}
	\max_{w \in [q]^n } \max_{ |\Lambda| = L  }  \min_{c \in \Lambda} \ip{\vphi(w)}{\vphi(c)} &\leq \frac{1}{L} \max_{w \in [q]^n } \max_{|\Lambda| = L } \sum_{c \in \Lambda} \ip{\vphi(w)}{\vphi(c)} \\
&= \frac{1}{L} \max_{w \in [q]^n } \max_{x \in \Sigma_{L}} \vphi(w)^T \Phi x\\
&\leq \frac{1}{L} \max_{w \in [q]^n} \|\vphi(w)\|_\infty \max_{x \in \Sigma_{L}} \|\Phi x \|_1\\
&= \frac{1}{L} \max_{x \in \Sigma_{L}} \|\Phi x \|_1.
\end{align*}
Thus it suffices to bound the last line by $(q-1)n\eps$.
\end{proof}

\subsection{Aside: the Restricted Isometry Property}
\label{ssec:RIP}
A matrix $A$ has the \em Restricted Isometry Property \em (RIP) if, for some constant $\delta$ and sparsity level $s$, 
\[ (1 - \delta) \|x \|_2^2 \leq \|A x \|_2^2 \leq (1 + \delta) \|x\|_2^2 \]
for all $s$-sparse vectors $x$.  The best constant $\delta = \delta( A, k )$ is called the \em Restricted Isometry Constant\em.   
  The RIP is an important quantity in compressed sensing, and much work has gone into understanding it.

CGV have shown that if $\frac{1}{\sqrt{n(q-1)}}\varphi(\mathcal{C})$ has the RIP with appropriate parameters, $\mathcal{C}$ is list decodable.  The proof that the RIP is a sufficient condition follows, after some computations, from an average-case Johnson bound.  While the average-case Johnson bound is interesting on its own, in this section we note that Lemma \ref{lem:sufficient} implies the sufficiency of the RIP immediately.
Indeed, by Cauchy-Schwarz,
\begin{align*}
	\frac{1}{L} \max_{x \in \Sigma_{L}} \|\Phi x\|_1 &\leq \frac{\sqrt{n(q-1)}}{L} \max_{x \in \Sigma_{L}} \|\Phi x\|_2 \\
&\leq \frac{\sqrt{n(q-1)}}{L}\inparen{\sqrt{n(q-1)} (1 + \delta) \max_{x \in \Sigma_L}\|x\|_2 }\\
&\leq \frac{ n(q-1)}{\sqrt{L}} (1 + \delta),
\end{align*}
where $\Phi = \varphi(\mathcal{C})$, and $\delta = \delta(\tilde{\Phi}, L)$ is the restricted isometry constant for $\tilde{\Phi} = \frac{1}{\sqrt{n(q-1)}} \Phi$ and sparsity $L$.
By Lemma \ref{lem:sufficient}, this implies that
\[ \frac{\delta + 1}{\sqrt{L}}  < \eps\]
also implies \eqref{eq:star}, and hence $(\bigrho, L-1)$-list decodability.
Setting $\delta = 1/2$, we may conclude the following statement:
\begin{quote}
For any code $\mathcal{C} \subset [q]^n$, if  $\frac{1}{\sqrt{n(q-1)}}\vphi(\mathcal{C})$ has the RIP with contant $1/2$ and sparsity level $L$, 
then $\mathcal{C}$ is $\inparen{ \inparen{1 - \nicefrac{1}{q}}\left(1  - \nicefrac{3}{2\sqrt{L}}\right), L-1}$-list decodable.
\end{quote}
This precisely recovers Theorem 11 from~\cite{che:2012}.


\section{A random linear code is list decodable}
\label{sec:tech}
We wish to show that, when $\Phi = \vphi(\mathcal{C})$ for a random linear code $\mathcal{C}$, \eqref{eq:doublestar} holds with high probability.  Thus, we need to bound $\max_{x \in \Sigma_L} \| \Phi x\|_1$.  We write
\begin{equation}\label{eq:decomp}
\max_{x \in \Sigma_L} \|\Phi x\|_1 \leq \max_{x \in \Sigma_L} \EE \| \Phi x\|_1 + \max_{x \in \Sigma_L}\inabs{ \| \Phi x\|_1 - \EE\| \Phi x\|_1 },
\end{equation}
and we will bound each term separately.  
First, we observe that $\EE \| \Phi x\|_1$ is correct.  


\begin{lemma}\label{lem:randexp}
Let $\mathcal{C} \subset \F_q^n$ be a linear $q$-ary code with a random generator matrix.   Let $\Phi = \vphi(\mathcal{C})$ as above.  Then for any $x \in \Sigma_L$, 
\[ \frac{1}{L} \EE \| \Phi x\|_1 \leq \frac{n(q-1)}{\sqrt{L}}.\]
\end{lemma}
\begin{proof}
The proof is a straighforward consequence of \eqref{eq:expect}.  For any $x \in \Sigma_L$, we have
\begin{align*}
\EE \|\Phi x\|_1 &\leq \sqrt{n(q-1)} \EE \|\Phi x\|_2\\
&\leq \sqrt{n(q-1)} \inparen{ \EE \|\Phi x\|_2^2 }^{1/2} \\
&= n(q-1)\sqrt{L}
\end{align*}
using \eqref{eq:expect} and the fact that $\|x\|_2 = \sqrt{L}$.
\end{proof}

Next, we control the deviation of $\|\Phi x\|_1$ from $\EE \|\Phi x\|_1$, uniformly over $x \in \Sigma_L$.
We do not require the vectors $t_j$ be drawn uniformly at random anymore, so long as they are selected independently.
\begin{lemma}\label{lem:concentrate}
Let $\mathcal{C} \subset \F_q^n$ be $q$-ary linear code, so that the columns $t_1, \ldots, t_n$ of the generator matrix are independent.  
Then
\[ \frac{1}{L}\EE \max_{x \in \Sigma_L} \inabs{ \| \Phi x \|_1 - \EE \| \Phi x\|_1 } \leq \Cproblemma (q-1)\sqrt{n\ln(N)}\]
with probability $1 - 1/\poly(N)$, for an absolute constant $\Cproblemma$.
\end{lemma}

\begin{remark}
As noted above, we do not make any assumptions on the distribution of the vectors $t_1, \ldots, t_n$, other than that they are chosen independently.  In fact, we do not even require the code to be linear---it is enough for the vectors $v_i = (c(i))_{c \in \mathcal{C}} \in [q]^N$ to be independent.  However, as we only consider linear codes in this work, we stick with our statement in order to keep the notation consistent.
\end{remark}

As a warm-up to the proof, which involves a few too many symbols, consider first the case when $q = 2$, and suppose that we wish to succeed with constant probability.  Then the rows $\Arow_t$ of $\Phi$ are rows of the Hadamard matrix, chosen independently.
By standard symmetrization and comparison arguments (made precise below), it suffices to bound
\begin{align*}
\frac{1}{L} \EE \max_{x \in \Sigma_L} \sum_{t \in T}  g_t \ip{\Arow_t}{x}  
& = \frac{1}{L} \EE \max_{x \in \Sigma_L}  \ip{g}{\Phi x}  \\
&\leq 
\EE \max_{x \in B^N_1} \ip{ g }{ \Phi x } \\
& = \EE \max_{y \in \Phi B_1^N}  \ip{g}{y},
\end{align*}
where above $g = (g_1, g_2, \ldots, g_n)$ is a vector of i.i.d. standard normal random variables, and $B_1^N$ denotes the $\ell_1$ ball in $\R^N$.
The last line is the mean width of $\Phi B_1^N$, which is a polytope contained in the convex hull of $\pm \vphi(c)$ for $c \in \mathcal{C}$, (that is, the columns of $\Phi$ and their opposites).  
So, using estimates for Gaussian random variables~\cite[Eq. (3.13)]{lt}, 
\begin{align*}
\EE \max_{y \in \Phi B_1^N} \ip{g}{y} &= \EE \max_{c \in \mathcal{C}} \ip{g}{\vphi(c)} \\
&\leq 3 \sqrt{\log|\mathcal{C}|} \sqrt{ \EE \ip{g}{\vphi(c)}^2 } \\
&= 3\|c\|_2 \sqrt{\log(N)} \\
&= 3\sqrt{n \log(N)}
\end{align*}
which is what we wanted.

For general $q$ and failure probability $o(1)$, there is slightly more notation, but the proof idea is the same.
We will need the following bound on moments of maxima of Gaussian random variables:
\begin{lemma}\label{lem:variance}
Let $X_1, \ldots, X_N$ be standard normal random variables (not necessarily independent).  Then
\[ \inparen{\EE \max_{i \leq N} |X_i|^p}^{1/p} \leq \Cvar N^{1/p} \sqrt{p} \]
for some absolute constant $\Cvar$.
\end{lemma}
\begin{proof} 
Let $Z = \max_{i \leq N} |X_i|$.  
Then
\[ \PR{ Z > s } \leq N \exp( -s^2/2  )\]
for $s \geq 1$. 
Integrating,
\begin{align*}
\EE |Z|^p &= \int \PR{ Z^p > s}\,ds \\
&= \int \PR{ Z^p >  t^p } p t^{p-1}\,dt\\
&\leq 1 + N \int_1^\infty \exp(-t^2/2) pt^{p-1}\,dt\\
&\leq 1 + Np 2^{p/2} \Gamma(p/2)\\
&\leq 1 + (Np) \left( p^{p/2}\right).
\end{align*}
Thus, 
\[ \inparen{\EE |Z|^p}^{1/p} \leq \Cvar N^{1/p} \sqrt{p}.\]
for some absolute constant $\Cvar$.
\end{proof}

Now we may prove the lemma.
\begin{proof}[of Lemma \ref{lem:concentrate}]
We recall the notation from the facts in Section \ref{sec:setup}: the rows of $\Phi$ are $\Arow_{t,\alpha}$ for $t \in T$, where $T$ is a random multiset of size $n$ with elements chosen independently from $\F_q^d$, and $\alpha \in \F_q^*$.


To control the largest deviation of $\|\Phi x\|_1$ from its expectation, we will control the $p^{th}$ moments of this deviation---eventually we will choose $p \sim \ln(N)$.  By a symmetrization argument followed by a comparison principle (Lemma 6.3 and Equation (4.8), respectively,  in~\cite{lt}),
for any $p \geq 1$,
\begin{align}
& \EE \max_{x \in \Sigma_L} | \|\Phi x\|_1 - \EE \|\Phi x\|_1 |^p \notag \\
& \qquad = \EE \max_{x \in \Sigma_L} \inabs{ \sum_{t \in T} \sum_{\alpha \in \F_q^*} \left( |\ip{\Arow_{t,\alpha}}{x}| - \EE|\ip{\Arow_{t,\alpha}}{x}|  \right) }^p
\notag\\
&\qquad \leq \Clt \EE_T  \EE_g  \max_{x \in \Sigma_L} \left| \sum_{t \in T} g_t \sum_{\alpha \in \F_q^*} |\ip{\Arow_{t,\alpha}}{x}| \right|^p 
\notag\\
&\qquad \leq \Clt \EE_T \EE_g \max_{x \in \Sigma_L} \left| (q-1) \max_{\alpha \in \F_q^*} \sum_{t \in T} g_t |\ip{\Arow_{t,\alpha}}{x}| \right|^p 
\notag\\
&\qquad \leq \Clt 4^p(q-1)^p \EE_T \EE_g \max_{x \in \Sigma_L} \max_{\alpha \in \F_q^*} \left| \sum_{t \in T} g_t \ip{\Arow_{t,\alpha}}{x} \right|^p ,
\label{eq:leftoff}
\end{align}
where the $g_t$ are i.i.d. standard normal random variables, and we dropped the absolute values 
at the cost of a factor of four by a contraction principle (see Cor. 3.17 in~\cite{lt}).
Above, we used the independence of the vectors $\Arow_{t,\alpha}$ for a fixed $\alpha$ to apply the symmetrization.

For fixed $\alpha$, let $\Phi_\alpha$ denote $\Phi$ restricted to the rows $\Arow_{t, \alpha}$ that are indexed by $\alpha$.  
Similarly, for a column $\vphi(c)$ of $\Phi$, let $\vphi(c)_\alpha$ denote the restriction of that column to the rows indexed by $\alpha$.
Conditioning on $T$ and fixing $\alpha \in \F_q^*$, let
\[X(x,\alpha) := \sum_{t \in T} g_t \ip{\Arow_{t, \alpha}}{x} = \ip{ g }{ \Phi_\alpha x }. \]
Let $B_1^N$ denote the $\ell_1$ ball in $\R^N$.  
Since $\Sigma_L \subset LB_1^N$, we have 
\[\Phi_\alpha (\Sigma_L) \subset L \Phi_\alpha(B_1^N) = \conv\{ \pm L \vphi(c)_\alpha \suchthat c \in \mathcal{C} \}.\] 
Thus, we have 
\begin{align}
&\EE_g \max_{x \in \Sigma_L} \max_{ \alpha \in \F_q^*} |X(x,\alpha)|^p \notag \\
&\qquad = \EE_g \max_{y \in \Phi_\alpha \Sigma_L}\max_{\alpha \in \F_q^*} |\ip{g}{y}|^p \notag\\
&\qquad \leq L^p\, \EE_g \max_{ \pm c \in \mathcal{C}  } \max_{\alpha \in \F_q^*} |\ip{g}{\vphi(c)_\alpha}|^p \label{eq:leftoff2},
\end{align}
using the fact that $\max_{x \in \conv(S)} F(x) = \max_{x \in S} F(x)$ for any convex function $F$.
Using Lemma \ref{lem:variance}, and the fact that $\ip{g}{\vphi(c)_\alpha}$ is Gaussian with variance $\|\vphi(c)_\alpha\|_2^2 = n$,
\begin{align}
&L^p\, \EE_g \max_{\pm c \in \mathcal{C}} \max_{\alpha \in \F_q^*} |\ip{g}{\vphi(c)_\alpha}|^p \notag\\
&\qquad \qquad \leq \inparen{ \Cvar\, L\, \sqrt{np} (2N(q-1))^{1/p} }^p. \label{eq:leftoff3}
\end{align}

Together, \eqref{eq:leftoff}, \eqref{eq:leftoff2}, and \eqref{eq:leftoff3} imply
\begin{align*}
&\EE \max_{x \in \Sigma_L} | \|\Phi x\|_1 - \EE \|\Phi x\|_1 |^p\\
&\qquad \leq \Clt 4^p(q-1)^p \EE_T \left(  \Cvar L \sqrt{np} (2N(q-1))^{1/p}\right)^p \\
&\qquad \leq \left( 4 \Clt^{1/p} \Cvar (q-1)^{(1 + 1/p)} L \sqrt{np} (2N)^{1/p} \right)^p\\
&\qquad=: Q(p)^p.
\end{align*}
Finally,  we set $p = \ln(N)$, so we have
\[ Q(\ln(N)) \leq \Ctmp (q-1) L\sqrt{n \ln(N)},\]
for an another constant $\Ctmp$.
Then Markov's inequality implies
\[ \PR{ \max_{x \in \Sigma_L} |\|\Phi x\|_1 - \EE \|\Phi x\|_1 | > e Q(\ln(N)) } \leq \frac{1}{N}.\]
We conclude that with probability at least $1 - o(1)$, 
\[ \frac{1}{L} \max_{x \in \Sigma_L} \inabs{ \|\Phi x\|_1 - \EE \|\Phi x\|_1 } \leq \Cproblemma(q-1)\sqrt{n \ln(N) },\]
for $\Cproblemma = e\Ctmp$.
\end{proof}
Now we may prove Theorem \ref{thm:mainthm}.
\begin{proof}[of Theorem \ref{thm:mainthm}]
Lemmas \ref{lem:randexp} and \ref{lem:concentrate}, along with \eqref{eq:decomp}, imply that
\[ \frac{1}{L} \max_{x \in \Sigma_L} \|\Phi x\|_1 \leq \frac{n(q-1)}{\sqrt{L}} + \Cproblemma(q-1)\sqrt{n\ln(N)} \]
with probability $1 - o(1)$.  Thus, if
\begin{equation}\label{eq:triplestar}
	(q-1) \left( \frac{n}{\sqrt{L}} + \Cproblemma\sqrt{n\ln(N)}\right)  < (q-1)n\eps
\end{equation}
holds, the condition \eqref{eq:doublestar} also holds with probability $1 - o(1)$.
Setting $L = \left(\nicefrac{2}{\eps}\right)^2$ and $n = \frac{4 \Cproblemma^2 \ln(N)}{\eps^2}$ 
satisfies \eqref{eq:triplestar},
so Lemma \ref{lem:sufficient} implies that $\mathcal{C}$ is $(\bigrho, 4/\eps^2)$-list decodable,
with $k$ equal to
\[ \log_q(N)  =  \frac{n \eps^2}{(2\Cproblemma)^2 \ln(q) }.\] 
With the remarks from Section \ref{sec:setup} following the definition of random linear codes, this concludes the proof.
\end{proof}

\section{Generalizations}
\label{sec:gen}
In this section, we show that our approach above applies not just to random linear codes, but to many ensembles. 
In our proof of Theorem \ref{thm:mainthm}, 
we required only that the expectation of $\|\Phi x\|_1$ be about right, and that the columns of the generator matrix were chosen independently, so that Lemma \ref{lem:concentrate} implies concentration.  The fact that $\|\Phi x\|_1$ was about right followed from the condition \eqref{eq:expect}, which required that, within sets $\Lambda \subset \mathcal{C}$ of size $L$, the average pairwise distance is, in expectation, large.  We formalize this observation in the following lemma, which can be substituted for Lemma \ref{lem:randexp}.
\begin{lemma}\label{lem:avgdist}
Let $\mathcal{C} = \{c_1,\ldots, c_N\} \subset [q]^n$ be a (not necessarily uniformly) random code so that for any $\Lambda \subset [N]$ with $|\Lambda| = L$, 
\begin{equation}\label{eq:reallyneed}
 \frac{1}{{L \choose 2}} \EE \sum_{ i < j \in \Lambda } d(c_i, c_j) \geq \inparen{1 - \nicefrac{1}{q}}\inparen{ 1 - \eta }.
\end{equation}
Then for all $x \in \Sigma_L$,
\[ \frac{1}{L} \EE \| \vphi(\mathcal{C})x \|_1 \leq n(q-1) \sqrt{ \frac{1}{L} + \frac{2 \eta {L \choose 2}}{L^2} }.\]
\end{lemma}
\begin{proof}
Fix $x \in \Sigma_L$, and let $\Lambda$ denote the support of $x$.  Then, using \eqref{eq:ip},
\begin{align*}
\frac{1}{L} \EE \| \vphi(\mathcal{C})x \|_1 &\leq \frac{ \sqrt{n(q-1)}}{L} \inparen{ \EE \| \vphi(\mathcal{C}) \|_2^2 }^{1/2}\\
&= \frac{ \sqrt{n(q-1)}}{L} \inparen{ \EE \sum_{i,j \in \Lambda} \ip{ \vphi(c_i)}{\vphi(c_j)}  }^{1/2}\\
&= \frac{ \sqrt{n(q-1)}}{L} \inparen{ \EE \sum_{i,j \in \Lambda} (q-1)n - qn\,d(c_i,c_j) }^{1/2}\\
&\leq \frac{ \sqrt{n(q-1)}}{L} \inparen{ L(q-1)n + 2{L \choose 2} n(q-1)\eta }^{1/2}\\
&= n(q-1) \sqrt{ \frac{1}{L} + \frac{2\eta {L \choose 2}}{L^2} },
\end{align*}
as claimed.
\end{proof}

Thus, we may prove a statement analogous to Theorem \ref{thm:mainthm} about any distribution on linear codes whose generator matix has independent columns, which satisfies \eqref{eq:reallyneed}. Where might we find such distributions?  Notice that if the expectation is removed, \eqref{eq:reallyneed} is precisely the hypothesis of the average case Johnson bound (Theorem 8 in~\cite{che:2012}), and so any code $\mathcal{C}$ to which the average case Johnson bound applies attains \eqref{eq:reallyneed}.  However, such a code $\mathcal{C}$ might have substantially suboptimal rate---we can improve the rate, and still satisfy \eqref{eq:reallyneed}, by forming generator matrix for a new code $\mathcal{E}$ from a random set of columns of the generator matrix of $\mathcal{C}$. 
\begin{definition}\label{def:subsample}
Fix a code $\mathcal{C} \subset [q]^{n'}$, and define an ensemble $\mathcal{E} = \mathcal{E}(\mathcal{C}) \subset [q]^{n}$ as follows.  To draw $\mathcal{E}$, choose a random multiset $T$ of size $n$ by drawing elements of $[n']$ independently with replacement.  Then let
\[ \mathcal{E} = \inset{ (x_{t_1}, \ldots, x_{t_{n}} ) \suchthat x \in \mathcal{C} }.\]
\end{definition}
\begin{remark}
We may think of the operation in Definition \ref{def:subsample} as randomly puncturing $\mathcal{C}$.  This is not quite correct, because the vectors $t_j$ are sampled with replacement, but it is correct in spirit.  In particular, all of the results that follow would hold if we retained each coordinate in $[n']$ independently with probability $n/n'$, and this would indeed be a punctured code, with \em expected \em length $n$.  Ignoring these technicalities, we will refer below to the codes of Definition \ref{def:subsample} as ``randomly punctured codes."
\end{remark}
Replacing Lemma \ref{lem:randexp} with Lemma \ref{lem:avgdist} in the proof of Theorem \ref{thm:mainthm}
immediately 
implies that randomly punctured codes
are list decodable with high probability, if the original code $\mathcal{C}$ has good average distance.  
\begin{cor}\label{cor:subsample}
Let $\mathcal{C} = \{c_1,\ldots,c_N\} \subset \F_q^{n'}$ be any linear code with 
\[ \frac{1}{{L \choose 2}}  \sum_{i < j \in \Lambda} d(c_i,c_j) \geq \inparen{1 - \frac{1}{q} }(1 - \eta) \]
for all sets $\Lambda \subset [N]$ of size $L$.  
Set
\[\eps^2 := 4\inparen{ \frac{1}{L}  + \eta\inparen{ 1 - \frac{1}{L} } }.\]
There is some $R = \Omega(\eps^2)$ so that if
 $\mathcal{E} = \mathcal{E}(\mathcal{C})$ is as in Definition \ref{def:subsample} with rate $R$, 
then $\mathcal{E}$ is $\inparen{ \bigrho, L-1 }$-list decodable with probability $1 - o(1)$.
\end{cor}
Theorem 8 in~\cite{che:2012} implies that if $\mathcal{C}$ is as in the statement of Corollary \ref{cor:subsample}, then $\mathcal{C}$ itself is $\inparen{ \bigrho, O(1/\eps^2)}$-list decodable, for $\eps$ as above.  
Thus, Corollary \ref{cor:subsample} implies that $\mathcal{E}(\mathcal{C})$ has the same list decodability properties as $\mathcal{C}$, but perhaps a much better rate.

As a example of this construction, consider the family of (binary) degree $r$ Reed-Muller codes, $\RM(r,m) \subset \F_2^m$.
$\RM(r,m)$ can be viewed as the set of degree $r$, $m$-variate polynomials over $\F_2$. It is easily checked that $\RM(r,m)$ is a linear code of dimension $k = 1 + {m \choose 1} + {m \choose 2} + \cdots + {m \choose r}$ and minimum relative distance $2^{-r}$.  
The resulting ensemble $\mathcal{E} = \mathcal{E}(\RM(r,m))$ is a natural class of codes: decoding $\mathcal{E}$ is equivalent to learning a degree $r$ polynomial over $\F_2^m$ from random samples, in the presence of (worst case) noise. 

We cannot hope for short list sizes in this case, but we can hope for nontrivial ones.  Kaufman, Lovett, and Porat~\cite{klp} have given tight asymptotic bounds on the list sizes for $\RM(r,m)$ for all radii, and in particular have shown that $\RM(r,m)$ is list decodable up to $1/2 - \eps$ with list sizes on the order of $\eps^{\Theta_r(m^{r-1})}$.  As $|\RM(r,m)|$ is exponential in $m^r$, this is a nontrivial bound.  We will show that randomly punctured Reed-Muller codes, with rate $\Omega(\eps^2)$, have basically the same list decoding parameters as their un-punctured progenitors.
\begin{proposition} Let $\mathcal{E} = \mathcal{E}(\RM(r,m))$ be as in Definition \ref{def:subsample}, with rate $O(\eps^2)$.
Then $\mathcal{E}$ is $\inparen{\nicefrac{1}{2}(1-\eps), L(\eps)}$-list decodable with probability $1 - o(1)$, where
\[ L(\eps) = \inparen{ \frac{1}{\eps} }^{O_r(m^{r-1})},\]
where $O_r$ hides constants depending only on $r$.
\end{proposition}
\begin{proof}
We aim to find $\eta$ so that \eqref{eq:reallyneed} is satisfied.  As usual, let $N = |\RM(r,m)|$.
We borrow a computation from the proof of Lemma 6 in~\cite{che:2012}.  
Let $A = A(\eps)$ be the number of codewords of $\RM(r,m)$ with relative weight at most $\nicefrac{1}{2}(1 - \eps^2)$.  
Let $L = A/\eps^2$ and choose a set $\Lambda \subset [N]$ of size $L$.  By linearity, for each codeword $c_i$ with $i \in \Lambda$, there are at most $A - 1$ codewords $c_j$ within $\nicefrac{1}{2}(1 - \eps^2)$ of $c_i$, out of $L - 1$ choices for $c_j$.  Thus, the sum of the relative distances over $j \neq i$ is at most $(L - A) \cdot \nicefrac{1}{2}(1 - \eps^2)$.  This implies
\begin{align*}
\frac{1}{{L \choose 2}} \sum_{i < j \in \Lambda} d(c_i,c_j) &\geq \frac{L - A}{L - 1} \inparen{ \frac{1}{2} (1-\eps^2)}\\
&= \inparen{ 1 - \frac{A - 1}{ L - 1} }\inparen{  \frac{1}{2}(1 - \eps^2)}\\
&\geq \frac{1}{2} \inparen{ 1 - \eps^2 - \frac{A-1}{L -1 } }\\
&= \frac{1}{2} \inparen{ 1 - O(\eps^2) },
\end{align*}
using the choice $L = A/\eps^2$ in the final line.
Thus, in Corollary \ref{cor:subsample}, we may take $\eta = O(\eps^2)$.

We conclude that the randomly punctured code $\mathcal{E}(\RM(r,m))$ of rate $O(\eps^2)$ is 
$\inparen{\nicefrac{1}{2}(1-\eps), L -1}$ list decodable, 
with list size $L$ on the order of $A/\eps^2$.  
It remains to estimate $A = A(\eps)$.
It is shown in~\cite{klp} that 
\[ A = A(\eps) = \inparen{\frac{1}{\eps}}^{\Theta_r(m^{r-1})},\]
which finishes the proof.
\end{proof}


Another popular ensemble of linear codes is the \em Wozencraft ensemble \em \cite{j72,w73}, which encodes an element $x \in \F_{q^k}$ as $(x, \alpha_1 x, \alpha_2 x, \ldots, \alpha_r x)$ for uniformly random $\alpha_j \in \F_{2^k}$.  In this case, the symbols within a codeword are not all independent, so Lemma \ref{lem:avgdist} does not apply.  However, the techniques above extend immediately to imply that a code from this ensemble (with $r \sim k/\eps^2)$ is $\inparen{\bigrho, O(1/\eps)}$-list decodable with rate $\eps^2/k$.  (Previously, the only known result about the list decodability of the Wozencraft ensemble follows from the Johnson bound, which implies a rate on the order of $\eps^4$ for the same radius, so for very small $\eps$ this is better).  It would be interesting to see if this argument could be modified to obtain constant rate for the Wozencraft ensemble, or for other ensembles of linear codes.

\section{Conclusion}
We have shown that a random linear code of rate $\Omega\inparen{\frac{\eps^2}{\log(q)}}$ is $\inparen{ \bigrho, O(1/\eps)}$-list decodable with probability $1 - o(1)$.  
Our result improves the results of~\cite{che:2012} in three ways.  First, we remove the logarithmic dependence on $\eps$ in the rate, achieving the optimal  dependence on $\eps$.  Second, it improves the dependence on $q$, from $1/\log^4(q)$ to $1/\log(q)$.  Finally, we show that list decodability holds with probability $1 - o(1)$, rather than with constant probability.  
Our result is the first to establish the existence of optimally list decodable $q$-ary linear codes for this parameter regime for general $q$.
As an added benefit, our proof is relatively short and straightforward.
To illustrate the applicability of our argument, we showed that in fact our techniques apply to many ensembles of random codes, including randomly punctured codes.  As an example, we considered Reed-Muller codes, and showed that they retain their combinatorial list decoding properties with high probability when randomly punctured down to constant rate. 

\section*{Acknowledgements}
I thank Atri Rudra and Martin Strauss for very helpful conversations.

\bibliographystyle{alpha}
\bibliography{codes.bib}
\end{document}